\documentclass[journal]{IEEEtran}

\usepackage{amsmath,amssymb,mathtools}
\usepackage{booktabs,array}
\usepackage{enumitem}
\usepackage{graphicx}
\usepackage{xcolor}
\usepackage{hyphenat}

\usepackage{tikz}
\usetikzlibrary{positioning,arrows.meta,backgrounds,calc}
\definecolor{full}{HTML}{2E8B57}     %
\definecolor{charging}{HTML}{E8A23D} %
\definecolor{peakcol}{HTML}{D1495B}  %
\definecolor{offcol}{HTML}{2E6E8E}   %
\newcommand{\batt}[4]{%
  \begin{scope}[shift={(#1,#2)}]
    \fill[#4] (0.025,0.025) rectangle ({0.025+0.35*#3},0.205);
    \draw[line width=0.4pt] (0,0) rectangle (0.40,0.23);
    \fill[black!75] (0.40,0.075) rectangle (0.45,0.155);
  \end{scope}}
\newcommand{\battchg}[2]{%
  \begin{scope}[shift={(#1,#2)}]
    \fill[charging] (0.025,0.025) rectangle (0.20,0.205);
    \draw[line width=0.4pt] (0,0) rectangle (0.40,0.23);
    \fill[black!75] (0.40,0.075) rectangle (0.45,0.155);
    \draw[line width=0.55pt, black] (0.255,0.185)--(0.165,0.105)--(0.215,0.105)--(0.135,0.030);
  \end{scope}}
\newcommand{\caricon}[2]{%
  \begin{scope}[shift={(#1,#2)}, line width=0.5pt]
    \draw[fill=black!5] (-0.34,0)--(-0.19,0)--(-0.10,0.17)--(0.12,0.17)
        --(0.20,0)--(0.34,0)--(0.34,-0.10)--(-0.34,-0.10)--cycle;
    \fill[black!60] (-0.19,-0.12) circle (0.058);
    \fill[black!60] (0.19,-0.12) circle (0.058);
  \end{scope}}

\usepackage{microtype}
\newtheorem{proposition}{Proposition}

\usepackage[hidelinks]{hyperref}

\begin{document}

\title{Scalable No-Stockout Charging Scheduling for Battery Swapping Under
Time-of-Use Prices}

\author{Eunbin~Cho, Junki~Cho, Hakjin~Lee, Jaehoon~Sim, and Junghoon~Seo%
\thanks{Eunbin Cho is with the School of Industrial and Systems Engineering,
Georgia Institute of Technology, Atlanta, GA 30332 USA (e-mail:
michelle98756@gmail.com). This work was performed while Eunbin Cho was an intern at
PIT IN. Junki Cho, Hakjin Lee, Jaehoon Sim, and Junghoon Seo are with PIT IN
Corp., Anyang, Gyeonggi-do 14088, Republic of Korea (e-mail: hj@pitin-ev.com;
jk@pitin-ev.com; simjeh@pitin-ev.com; sjh@pitin-ev.com).}%
\thanks{Corresponding author: Junghoon Seo (e-mail: sjh@pitin-ev.com).}}

\maketitle

\begin{abstract}
A battery-swapping station must provide every arriving vehicle with a charged battery while minimizing the time-of-use cost of recharging returned units. Coordinating heterogeneous compatibility, vehicle-specific return times, and finite charger capacity requires service-aware recharge decisions across the planning horizon. We formulate a per-battery mixed-integer linear program that captures these operational features under a hard no-stockout constraint and derive a provably equivalent reduced form with fewer explicit binary variables. In the synthetic scaling study, a price-guided battery-path heuristic returned a full-service schedule for every instance; regime-level median solve times ranged from $0.24$ to $8.0$ seconds. Its median cost premiums were $7$--$8\%$ over certified reference costs for small- and medium-scale instances, and its certified ex post optimality-gap upper bounds were $9$--$12\%$ for large- and extra-large-scale instances. For each operational baseline, the certified reference schedules reduced charging-energy cost by $50$--$60\%$ on instances that the baseline fully served and for which a certified reference was available. In a 30-day replay of $1{,}002$ swaps recorded at a commercial station, the reduced-model and heuristic rolling controllers served every swap and reduced charging-energy cost by approximately $50\%$ relative to immediate charging.
\end{abstract}


\section{Introduction}
\label{sec:introduction}

\IEEEPARstart{A}{s} transportation electrifies, charging is a major operating cost for commercial electric-vehicle (EV) fleets. Connected commercial and autonomous-mobility fleets can communicate planned vehicle arrivals and expected battery-return times through their fleet-management systems. Battery swapping rapidly returns a vehicle to service and separates its operating time from battery charging: the station exchanges a depleted battery for a charged one and recharges the returned unit later~\cite{schneider2018,widrick2018}. Under time-of-use (TOU) pricing, this separation creates opportunities to shift recharging to low-price periods~\cite{tan2019,sarker2015}. The scheduling challenge is to capture these savings while preserving a charged battery for every scheduled arrival.

The recharge schedule determines both a swap station's cost and its reliability, and the literature offers two complementary perspectives. Cost-oriented models schedule recharging against time-varying prices using per-slot ready-battery counts~\cite{tan2019,zeng2025}, procure low-price energy for arbitrage~\cite{sarker2015}, modulate charging power~\cite{deng2026}, or learn schedules with average-service targets or service penalties~\cite{chen2024time,li2025cdrl}. Availability-oriented models size inventory and coordinate holding and recharging decisions to meet uncertain demand~\cite{schneider2018,widrick2018,sun2019}, representing service through replenishment policies or risk measures over interchangeable battery pools~\cite{han2025,shalaby2023}.

Each recharge jointly determines charging cost and the time at which a physical battery becomes ready for its next service. A unified finite-horizon formulation must therefore represent battery identity, event-specific compatibility and return times, finite charger capacity, a hard per-event no-stockout constraint, and certifiable TOU-cost minimization. The formulation developed here integrates these elements by making each service decision create a recharge job and linking that job to a specific battery, compatible charger, and start time across the horizon.

This paper formulates the problem as a program over individual batteries, determining cost and service jointly. The per-battery representation preserves physical identity through event-specific return times, compatibility, and repeated reuse. A hard constraint enforces no-stockout service for every event in the deterministic planning input. At each decision epoch, the station re-optimizes using the arrivals and expected return times communicated by the fleet-management system.

The exact program has two mixed-integer linear programming (MILP) forms. Model~F provides a direct transcription of the station dynamics, and Model~R preserves its optimum with fewer explicit binary variables (Section~\ref{sec:methods}). Model~R extends exact certification to larger station-scale instances, while a column-generation-based price-guided battery-path heuristic (P-BPD) returns full-service schedules across the evaluated range (Section~\ref{sec:price-guided-bpd}). The experiments quantify solution quality, service, and computational scaling (Section~\ref{sec:experiments}).

The \emph{contributions} of this paper are as follows:
\begin{itemize}[topsep=2pt, partopsep=0pt]
\item \textbf{An exact per-battery model and its reduction.} A single exact program unifies TOU cost minimization with a hard per-event no-stockout requirement, and an equivalent reduction concentrates the binary decisions in service assignments and charging starts.%
\item \textbf{Certified scaling and fast battery-path scheduling.} Model~R extends exact certification beyond the direct formulation's range, and P-BPD constructs feasible large-scale schedules with certified ex post optimality-gap bounds.%
\item \textbf{Commercial-station replay and demand scaling.} A 30-day station replay and a demand-scaling study quantify service, savings, and the computational crossover between Model~R and P-BPD.%
\end{itemize}

\section{Related Work}
\label{sec:related-work}

This problem connects three research areas: inventory control under service-level requirements, charge scheduling under TOU prices, and column generation for computational scalability. The proposed framework integrates their respective strengths---availability guarantees, a price-aware objective, and path-based solution machinery.

\paragraph{Battery inventory control under a service level}
For a swap station, charged-battery availability preserves the rapid service that motivates swapping. This availability problem corresponds to the lost-sales service level of classical inventory theory~\cite{zipkin2000} and is studied in transportation and energy operations research. The literature develops along two axes. On the planning axis, Schneider et al.~\cite{schneider2018} optimize joint battery purchasing and charging under demand uncertainty, and Sun et al.~\cite{sun2019} extend this approach to a price-aware fluid model. On the operational axis, Widrick et al.~\cite{widrick2018} derive monotone optimal charge-discharge policies via a Markov decision process, while Mao et al.~\cite{mao2023} obtain structural charging and replenishment policies under demand, supply, and price uncertainty. Heterogeneous compatibility extends this service requirement to individual physical batteries; the per-battery formulation supplies this event-level granularity while guaranteeing a ready compatible battery for every arrival.

\paragraph{Charge scheduling under time-of-use prices}
Charging cost is driven by TOU prices, which fall in the off-peak hours. Single-station cost minimization has been modeled through methods ranging from mixed-integer programs~\cite{tan2019,sarker2015} to a constrained Markov decision process~\cite{ko2022}. Zeng et al.~\cite{zeng2025} extend the problem to multiple stations coupled through a distribution network, modeling charging bays as explicit finite states. Subsequent work incorporates demand uncertainty and forecasts~\cite{nayak2024,pla2025demand}, controls charging power~\cite{deng2026}, and solves online through rolling-horizon dispatch~\cite{shalaby2022}, predictive control~\cite{casini2021,li2025degradation}, or reinforcement learning~\cite{chen2024time,li2025cdrl,shalaby2023,park2024,huang2024,jin2023,renga2024}. Representative cost-focused formulations express service through aggregate ready counts~\cite{tan2019,zeng2025}, soft penalties~\cite{chen2024time,li2025cdrl}, or chance constraints~\cite{han2025}. The present work adds physical-battery compatibility and a hard event-level service guarantee to a certifiable finite-horizon cost model.

\paragraph{Column generation and battery-path decomposition}
Column generation supports fast online re-solving by building a schedule from a compact, growing set of candidate columns~\cite{barnhart1998,desaulniers2005}. The technique underlies vehicle and transit scheduling with finite charging capacity~\cite{parmentier2023,devos2024} and TOU machine scheduling~\cite{ding2016,tian2024}, which shares the price-indexed temporal structure of the present cost subproblem. In the present setting, service decisions create the recharge jobs endogenously, and each job defines a battery's path from recharge to its next service. Related path columns support electric-vehicle service-network design~\cite{diao2026snd,diao2026swap}, while per-battery Lagrangian decomposition supports fleet charging~\cite{fotedar2025}. P-BPD specializes battery-path decomposition to single-station TOU recharge and benchmarks its schedules against certified reference solutions and dual-bound optimality-gap guarantees.

\section{System Model and Problem Definition}
\label{sec:system-model}
\subsection{Station operation}
\label{sec:background}

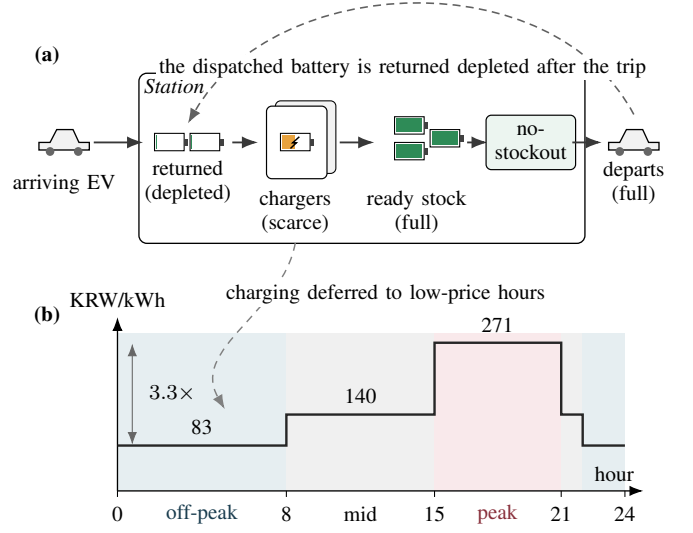
\begin{figure}[t]
\centering
\begin{tikzpicture}[>=Latex, font=\footnotesize, x=0.98cm, y=1cm,
  flow/.style={->, line width=0.7pt, black!78},
  ident/.style={->, densely dashed, line width=0.6pt, black!62},
  link/.style={->, dash pattern=on 3pt off 2.4pt, line width=0.55pt, black!50}]
\node[font=\footnotesize\bfseries] at (0.05,7.78){(a)};
\node[draw, rounded corners=3pt, line width=0.5pt,
      minimum width=5.9cm, minimum height=2.25cm] (stn) at (4.30,6.42){};
\node[anchor=north west, font=\itshape\footnotesize, inner sep=2pt]
      at (stn.north west){Station};
\caricon{0.28}{6.63}
\node[font=\footnotesize, align=center] at (0.28,6.08){arriving EV};
\draw[flow] (0.68,6.63) -- (1.35,6.63);
\batt{1.50}{6.52}{0.0}{full}
\batt{1.98}{6.52}{0.0}{full}
\node[font=\footnotesize, align=center] at (1.95,6.10){returned\\(depleted)};
\draw[flow] (2.55,6.63) -- (2.90,6.63);
\node[draw, rounded corners=2pt, line width=0.4pt, fill=black!5,
      minimum width=0.80cm, minimum height=0.96cm] at (3.54,6.70){};
\node[draw, rounded corners=2pt, line width=0.45pt, fill=white,
      minimum width=0.80cm, minimum height=0.96cm] (chg) at (3.42,6.63){};
\battchg{3.20}{6.52}
\node[font=\footnotesize, align=center, anchor=north] at (3.40,6.12){chargers\\(scarce)};
\draw[flow] (4.00,6.63) -- (4.47,6.63);
\batt{4.74}{6.74}{1.0}{full}
\batt{4.74}{6.40}{1.0}{full}
\batt{5.22}{6.57}{1.0}{full}
\node[font=\footnotesize, align=center, anchor=north] at (5.04,6.12){ready stock\\(full)};
\draw[flow] (5.72,6.63) -- (6.00,6.63);
\node[draw, rounded corners=2pt, line width=0.45pt, fill=full!9,
      minimum width=1.0cm, minimum height=0.72cm, align=center, font=\footnotesize]
      (gate) at (6.58,6.63){no-\\stockout};
\draw[flow] (7.14,6.63) -- (7.54,6.63);
\caricon{7.97}{6.63}
\node[font=\footnotesize, align=center] at (7.97,6.10){departs\\(full)};
\draw[ident] (7.97,7.00) to[out=118,in=60] (1.95,6.88);
\node[font=\footnotesize, align=center, fill=white, inner sep=1pt] at (4.85,7.62)
   {the dispatched battery is returned depleted after the trip};
\node[font=\footnotesize\bfseries] at (0.05,4.34){(b)};
\begin{scope}[yshift=1.32cm]
\begin{scope}[on background layer]
  \fill[offcol!12]  (1.00,0.70) rectangle (3.28,2.78);
  \fill[black!6]    (3.28,0.70) rectangle (5.28,2.78);
  \fill[peakcol!12] (5.28,0.70) rectangle (6.99,2.78);
  \fill[black!6]    (6.99,0.70) rectangle (7.28,2.78);
  \fill[offcol!12]  (7.28,0.70) rectangle (7.85,2.78);
\end{scope}
\draw[->, line width=0.5pt] (1.00,0.70) -- (8.10,0.70)
  node[above left,font=\footnotesize]{hour};
\draw[->, line width=0.5pt] (1.00,0.70) -- (1.00,3.00) node[above,font=\footnotesize]{KRW/kWh};
\draw[line width=0.9pt, black!85]
  (1.00,1.30)--(3.28,1.30)--(3.28,1.71)--(5.28,1.71)--(5.28,2.66)--(6.99,2.66)
  --(6.99,1.71)--(7.28,1.71)--(7.28,1.30)--(7.85,1.30);
\node[font=\footnotesize, anchor=south] at (2.14,1.32){83};
\node[font=\footnotesize, anchor=south] at (4.28,1.73){140};
\node[font=\footnotesize, anchor=south] at (6.13,2.68){271};
\foreach \h/\xx in {0/1.00,8/3.28,15/5.28,21/6.99,24/7.85}{
  \draw[line width=0.4pt] (\xx,0.70)--(\xx,0.62);
  \node[font=\footnotesize,anchor=north] at (\xx,0.62){\h};}
\node[font=\footnotesize, text=offcol!55!black] at (2.14,0.40){off-peak};
\node[font=\footnotesize] at (4.28,0.40){mid};
\node[font=\footnotesize, text=peakcol!65!black] at (6.13,0.40){peak};
\draw[<->, line width=0.4pt, black!60] (1.20,1.31) -- (1.20,2.65);
\node[font=\footnotesize, anchor=west] at (1.30,1.99){$3.3\times$};
\end{scope}
\draw[link] (3.40,5.30) to[out=-100,in=120] (2.45,3.10);
\node[font=\footnotesize, align=center, fill=white, inner sep=1pt] at (4.62,4.62)
   {charging deferred to low-price hours};
\end{tikzpicture}
\caption{(a)~Station operation and (b)~the time-of-use tariff. Glyphs: solid = full,
bolt = charging, outline = depleted.}
\label{fig:station}
\end{figure}

A battery-swap station serves a commercial electric-vehicle fleet whose productive time depends on vehicle availability. A brief exchange replaces a depleted battery with a fully charged one, returning the vehicle to service while charging proceeds separately.

The basic unit of this operation is the \emph{event}, a single customer swap. At an event, a vehicle arrives, receives a full battery, and departs. The battery returns depleted when the trip ends; the trip duration determines this \emph{return} epoch. The time spent away from the station is the trip-specific \emph{return delay}. The station then recharges the returned battery for its next issue (Fig.~\ref{fig:station}(a)).

The operation involves events, batteries, and chargers linked by two compatibility relations. Each event draws from its \emph{demand pool}, the batteries matching the vehicle's model and capacity. Each battery also belongs to a \emph{compatibility class} that identifies its compatible chargers. Thus, every service pairs an event with a battery from its demand pool and routes the returned battery to a charger of the same class.

Recharging capacity consists of a finite set of chargers, each reserved by one battery for its full planned recharge. Each recharge has a station-selected \emph{charge start} and a corresponding \emph{charge finish}. At each re-plan, its integer reserved duration is a deterministic input indexed by the event, battery, and charger, and $e_c$ gives the scheduled energy drawn in each occupied slot. The objective therefore gives the exact cost for these planning inputs. When the reserved duration is a charge-to-full upper bound, an early finish preserves schedule feasibility and makes the computed energy cost an upper planning estimate.

At any moment, a battery present at the station is in one of three states: \emph{ready} (full and available to serve), \emph{waiting} (returned but not yet charging), or \emph{charging}. An event's \emph{full stock} is the ready subset of its demand pool.
The station purchases electricity under a \emph{time-of-use tariff}, with lower off-peak prices. The station reduces cost by scheduling returned batteries in lower-price hours (Fig.~\ref{fig:station}(b)). The \emph{no-stockout rule} sets the service-aware timing of these recharges so that every arriving vehicle receives a battery from full stock. A rolling schedule implements this coordination.

In the connected-fleet setting, the fleet-management system supplies planned arrivals and expected return epochs over the look-ahead horizon to each re-plan. The re-plan assigns each event a battery and schedules the recharge associated with each in-horizon return on a compatible charger. It minimizes TOU charging cost subject to the no-stockout rule and charger capacity. The controller executes actions as they become due and refreshes the schedule at the next event using updated fleet information.

\subsection{Sets and parameters}

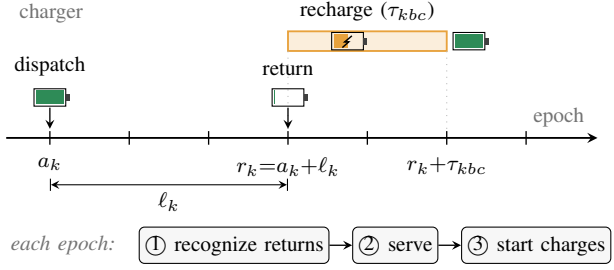
\begin{figure}[t]
\centering
\begin{tikzpicture}[>=stealth, font=\footnotesize, x=1.05cm, y=1.0cm]
  \node[anchor=west, font=\footnotesize, text=black!55] at (-0.50,1.65) {charger};
  \fill[charging!18] (3,1.15) rectangle (5,1.41);
  \draw[charging, line width=0.7pt] (3,1.15) rectangle (5,1.41);
  \battchg{3.55}{1.165}
  \node[font=\footnotesize] at (4.0,1.68) {recharge ($\tau_{kbc}$)};
  \batt{5.08}{1.165}{1}{full}
  \draw[gray!55, dotted] (3,0.12) -- (3,1.15);
  \draw[gray!55, dotted] (5,0.12) -- (5,1.15);
  \draw[->, line width=0.7pt] (-0.55,0) -- (6.9,0);
  \foreach \x in {0,1,2,3,4,5,6} \draw[line width=0.5pt] (\x,0.09)--(\x,-0.09);
  \node[below=1pt, font=\footnotesize] at (0,-0.10) {$a_k$};
  \node[below=1pt, font=\footnotesize] at (3,-0.10) {$r_k{=}a_k{+}\ell_k$};
  \node[below=1pt, font=\footnotesize] at (5,-0.10) {$r_k{+}\tau_{kbc}$};
  \node[anchor=east, font=\footnotesize, text=black!55] at (6.85,0.28) {epoch};
  \batt{-0.20}{0.42}{1}{full}
  \draw[->, line width=0.5pt] (0,0.40) -- (0,0.12);
  \node[above=8pt, font=\footnotesize] at (0,0.42) {dispatch};
  \batt{2.80}{0.42}{0}{full}
  \draw[->, line width=0.5pt] (3,0.40) -- (3,0.12);
  \node[above=8pt, font=\footnotesize] at (3,0.42) {return};
  \draw[|<->|, line width=0.4pt] (0,-0.62) -- (3,-0.62);
  \node[below=0pt, font=\footnotesize] at (1.5,-0.60) {$\ell_k$};
  \begin{scope}[yshift=-1.42cm]
    \node[font=\footnotesize\itshape, text=black!55, anchor=west, inner sep=0pt]
          (epochlabel) at (-0.50,0) {each epoch:};
    \node[draw, rounded corners=2pt, line width=0.4pt, fill=black!3, font=\footnotesize,
          inner xsep=2pt, inner ysep=3pt, minimum height=4.8mm,
          right=3mm of epochlabel] (s1) {\textcircled{1} recognize returns};
    \node[draw, rounded corners=2pt, line width=0.4pt, fill=black!3, font=\footnotesize,
          inner xsep=2pt, inner ysep=3pt, minimum height=4.8mm,
          right=3mm of s1] (s2) {\textcircled{2} serve};
    \node[draw, rounded corners=2pt, line width=0.4pt, fill=black!3, font=\footnotesize,
          inner xsep=2pt, inner ysep=3pt, minimum height=4.8mm,
          right=3mm of s2] (s3) {\textcircled{3} start charges};
    \draw[->, line width=0.5pt] (s1.east) -- (s2.west);
    \draw[->, line width=0.5pt] (s2.east) -- (s3.west);
  \end{scope}
\end{tikzpicture}
\caption{Within-epoch ordering for an immediate-start example ($s=r_k$): the battery is full at $a_k$, returns depleted at $r_k=a_k+\ell_k$, and recharges over $\tau_{kbc}$ slots. Glyphs as in Fig.~\ref{fig:station}.}
\label{fig:epoch-timing}
\end{figure}

Let $T_{\max}\in\mathbb Z_{++}$ denote the finite horizon length. Time is split into state epochs $\mathcal{T}^{\mathrm E}=\{0,1,\ldots,T_{\max}\}$ and operation slots $\mathcal{T}^{\mathrm S}=\{0,1,\ldots,T_{\max}-1\}$. Slot $t$ is the interval $[t,t+1)$. Epoch $t$ is the instant at the start of slot $t$, so an event time in $\mathcal{T}^{\mathrm S}$ names the epoch beginning that slot. At epoch $t$, battery returns and charging completions are recognized first. Swap events are then served from the pre-service full inventory, after which charging jobs selected at epoch $t$ start and occupy slot $t$.

Let $\mathcal{B}$, $\mathcal{C}$, and $\mathcal{K}$ be the finite sets of physical batteries, chargers, and customer swap events. Each event $k\in\mathcal{K}$ has an arrival epoch $a_k\in\mathcal{T}^{\mathrm S}$ and a return delay $\ell_k\in\mathbb{Z}_{++}$. Its return epoch is $r_k=a_k+\ell_k$, which may lie beyond the horizon. Let $\lambda_t\in\mathbb R_+$ denote the electricity price in slot $t$, and let $e_c\in\mathbb R_{++}$ be the scheduled energy charger $c$ uses in one occupied slot.

Each battery $b$ has a horizon-initial ready indicator $A_b^0\in\{0,1\}$, a demand pool $\mathrm{pool}_b$, and a compatibility class $\mathrm{class}_b$. The canonical formulation initializes ready batteries through $A_b^0$; the source-job augmentation in Section~\ref{sec:controllers} brings in-trip, waiting, and charging batteries into rolling re-plans. Each event $k$ requires demand pool $\mathrm{pool}_k$, and each charger $c$ has compatibility class $\mathrm{class}_c$. Define
\begin{align}
\mathcal{B}_k&=\{b\in\mathcal{B}:\mathrm{pool}_b=\mathrm{pool}_k\},\\
\mathcal{C}_b&=\{c\in\mathcal{C}:\mathrm{class}_c=\mathrm{class}_b\}.
\end{align}
If event $k$ issues physical battery $b$ and selects charger $c$, the returned battery's charging duration is $\tau_{kbc}\in\mathbb{Z}_{++}$. The feasible assignment and in-horizon return-job sets are
\begin{align}
\mathcal{Y}&=\{(k,b):k\in\mathcal{K},\ b\in\mathcal{B}_k\},\\
\mathcal{J}&=\{(k,b)\in\mathcal{Y}: r_k\le T_{\max}\}.
\end{align}
For $(k,b)\in\mathcal{J}$ and $c\in\mathcal{C}_b$, the feasible charging-start set is
\begin{equation}
\mathcal{S}_{kbc}=\{s\in\mathcal{T}^{\mathrm S}: r_k\le s,\ s+\tau_{kbc}\le T_{\max}\}.
\end{equation}
All summations over empty sets are interpreted as zero.
For every event $k$ with $r_k\le T_{\max}$, we assume the existence of a compatible pair $b\in\mathcal B_k$ and $c\in\mathcal C_b$ with $\mathcal S_{kbc}\ne\varnothing$; equivalently, every modeled return precedes $T_{\max}$ and admits an in-horizon charging start.

\section{Exact Formulation and Reduction}
\label{sec:methods}

This section presents two equivalent exact models for the offline no-stockout charge-scheduling problem: the direct five-family form~F (Section~\ref{sec:full-milp}) and its reduction~R (Section~\ref{sec:reduced}). Both models build on the finite-horizon, slot-indexed charging structure of prior battery-swapping scheduling~\cite{tan2019} by explicitly tracking service assignments and states for each physical battery.

\subsection{Decision variables and accounting functions}
\label{sec:variables}

Model~F uses five binary variable families: service assignment $y_{kb}$, charging start $\sigma_{kbcs}$, availability $A_{bt}$, waiting state $W_{bt}$, and charger occupancy $x_{kbct}$. Table~\ref{tab:notation} summarizes the formulation. Six derived quantities record issues, charging, completions, cumulative flows, and presence. At every feasible point, the constraints make $D_{bt}(y)$, $G_{bt}(x)$, and $I^+_{bt}(y)$ binary indicators; $\Gamma_{bu}(\sigma)$, $\rho_{bt}(y)$, and $\varsigma_{bt}(\sigma)$ retain their count interpretations.

\begin{table}[t]
\centering
\caption{Summary of notation.}
\label{tab:notation}
\footnotesize
\begin{tabular}{@{}lp{0.72\columnwidth}@{}}
\toprule
Symbol & Description \\
\midrule
\multicolumn{2}{@{}l}{\textit{Sets and indices}}\\
$\mathcal{T}^{\mathrm E}$ & epochs $\{0,\ldots,T_{\max}\}$ \\
$\mathcal{T}^{\mathrm S}$ & slots $\{0,\ldots,T_{\max}-1\}$ \\
$\mathcal{B}$, $b$ & batteries \\
$\mathcal{C}$, $c$ & chargers \\
$\mathcal{K}$, $k$ & swap events \\
$\mathcal{B}_k$ & batteries compatible with $k$ \\
$\mathcal{C}_b$ & chargers compatible with $b$ \\
$\mathcal{Y}$ & feasible assignments $(k,b)$ \\
$\mathcal{J}$ & in-horizon return jobs \\
$\mathcal{S}_{kbc}$ & feasible start slots, $(k,b)$ on $c$ \\
\midrule
\multicolumn{2}{@{}l}{\textit{Parameters}}\\
$T_{\max}$ & horizon length (slots) \\
$a_k$ & arrival epoch of $k$ \\
$\ell_k$ & return delay of $k$ \\
$r_k$ & return epoch of $k$ ($=a_k+\ell_k$) \\
$\lambda_t$ & price in slot $t$ \\
$e_c$ & energy $c$ draws per slot \\
$\tau_{kbc}$ & charging duration, $(k,b)$ on $c$ \\
$A_b^0$ & $1$ if $b$ is ready at the initial epoch \\
$\mathrm{pool}_b,\mathrm{pool}_k$ & demand pool of $b$ / $k$ \\
$\mathrm{class}_b,\mathrm{class}_c$ & compatibility class of $b$ / $c$ \\
\midrule
\multicolumn{2}{@{}l}{\textit{Decision variables (binary)}}\\
$y_{kb}$ & $=1$ if $k$ served by $b$ \\
$A_{bt}$ & $=1$ if $b$ full and ready before epoch $t$'s service \\
$W_{bt}$ & $=1$ if $b$ present, non-full, off charger \\
$\sigma_{kbcs}$ & $=1$ if job $(k,b)$ starts on $c$ at slot $s$ \\
$x_{kbct}$ & $=1$ if job $(k,b)$ occupies $c$ in slot $t$ \\
\midrule
\multicolumn{2}{@{}l}{\textit{Derived quantities}}\\
$D_{bt}(y)$ & same-epoch service count, $b$ \\
$G_{bt}(x)$ & charging-occupancy count, $b$ \\
$\Gamma_{bu}(\sigma)$ & charging completions of $b$ at epoch $u$ \\
$\rho_{bt}(y)$ & cumulative in-horizon returns of $b$ \\
$\varsigma_{bt}(\sigma)$ & cumulative charging starts of $b$ \\
$I^+_{bt}(y)$ & net post-service presence count of $b$ \\
\bottomrule
\end{tabular}
\end{table}

The battery cycle begins with an issue, as recorded by the same-epoch service count
\begin{equation}
D_{bt}(y)=\sum_{\substack{k\in\mathcal{K}:a_k=t,\ (k,b)\in\mathcal{Y}}}y_{kb},
\qquad b\in\mathcal{B},\ t\in\mathcal{T}^{\mathrm E},
\label{eq:D-def-paper}
\end{equation}
with $D_{b,T_{\max}}(y)=0$ following from $a_k\in\mathcal T^{\mathrm S}$. In every feasible solution, the serve row and binary availability imply $D_{bt}\in\{0,1\}$. After its return, a scheduled battery occupies a charger. Its charging-occupancy count is
\begin{equation}
G_{bt}(x)=
\begin{cases}
\displaystyle \sum_{\substack{(k,b)\in\mathcal{J}}}\sum_{c\in\mathcal{C}_b}x_{kbct}, & t\in\mathcal{T}^{\mathrm S},\\[0.7em]
0, & t=T_{\max}.
\end{cases}
\label{eq:G-def-paper}
\end{equation}
The constraints imply $G_{bt}(x)\in\{0,1\}$ at every feasible point. For $t\in\mathcal T^{\mathrm S}$, it indicates that battery $b$ occupies a charger in slot $t$, and $G_{b,T_{\max}}=0$ by definition.
Each recharge finishes $\tau_{kbc}$ slots after its start. The charging completions of battery $b$ at epoch $u$ are
\begin{equation}
\Gamma_{bu}(\sigma)=
\sum_{\substack{(k,b)\in\mathcal{J}}}
\sum_{c\in\mathcal{C}_b}
\sum_{\substack{s\in\mathcal{S}_{kbc}:s+\tau_{kbc}=u}}
\sigma_{kbcs}.
\label{eq:Gamma-def-paper}
\end{equation}
Two running totals track the cycle over the planning horizon. The cumulative in-horizon returns and cumulative charging starts are
\begin{align}
\rho_{bt}(y)&=\sum_{\substack{(k,b)\in\mathcal{Y}:r_k\le t}}y_{kb},
\label{eq:R-def-paper}\\
\varsigma_{bt}(\sigma)&=\sum_{\substack{(k,b)\in\mathcal{J}}}
\sum_{c\in\mathcal{C}_b}
\sum_{\substack{s\in\mathcal{S}_{kbc}:s\le t}}
\sigma_{kbcs}.
\label{eq:S-def-paper}
\end{align}
The post-service station-side presence is obtained by netting issues against initial stock and returns:
\begin{equation}
I^+_{bt}(y)=A_b^0+\rho_{bt}(y)-\sum_{\substack{k\in\mathcal{K}:(k,b)\in\mathcal{Y},\ a_k\le t}}y_{kb}.
\label{eq:Iplus-def-paper}
\end{equation}
At every feasible point, $I^+_{bt}(y)$ is integer-valued and satisfies $0\le I^+_{bt}(y)\le1$; it is therefore the post-service station-presence indicator for battery $b$ at epoch $t$.

\begin{figure*}[t]
\centering
\setlength{\fboxsep}{6pt}
\fbox{%
\begin{minipage}{\dimexpr\textwidth-2\fboxsep-2\fboxrule\relax}
\small
\setlength{\jot}{2pt}
\begin{subequations}
\label{prob:compact-milp-paper}
\begin{align}
\min_{y,A,W,\sigma,x}\quad
& \sum_{t\in\mathcal{T}^{\mathrm S}}\sum_{c\in\mathcal{C}}\lambda_t e_c
\sum_{\substack{(k,b)\in\mathcal{J}:c\in\mathcal{C}_b}}x_{kbct}
\label{eq:compact-obj-paper}\\
\mathrm{s.t.}\quad
& \sum_{b\in\mathcal{B}_k}y_{kb}=1,
&& \forall k\in\mathcal{K},
\label{eq:compact-assign-paper}\\
& D_{bt}(y)\le A_{bt},
&& \forall b\in\mathcal{B},\ t\in\mathcal{T}^{\mathrm S},
\label{eq:compact-serve-paper}\\
& A_{b0}=A_b^0,
&& \forall b\in\mathcal{B},
\label{eq:compact-init-paper}\\
& A_{b,t+1}=A_{bt}-D_{bt}(y)+\Gamma_{b,t+1}(\sigma),
&& \forall b\in\mathcal{B},\ t\in\mathcal{T}^{\mathrm S},
\label{eq:compact-Adyn-paper}\\
& \sum_{c\in\mathcal{C}_b}\sum_{s\in\mathcal{S}_{kbc}}\sigma_{kbcs}= y_{kb},
&& \forall (k,b)\in\mathcal{J},
\label{eq:compact-start-paper}\\
& W_{bt}=\rho_{bt}(y)-\varsigma_{bt}(\sigma),
&& \forall b\in\mathcal{B},\ t\in\mathcal{T}^{\mathrm E},
\label{eq:compact-W-paper}\\
& x_{kbct}=\sum_{\substack{s\in\mathcal{S}_{kbc}:s\le t<s+\tau_{kbc}}}\sigma_{kbcs},
&& \forall (k,b)\in\mathcal{J},\ c\in\mathcal{C}_b,\ t\in\mathcal{T}^{\mathrm S},
\label{eq:compact-duration-paper}\\
& \sum_{\substack{(k,b)\in\mathcal{J}:c\in\mathcal{C}_b}}x_{kbct}\le 1,
&& \forall c\in\mathcal{C},\ t\in\mathcal{T}^{\mathrm S},
\label{eq:compact-charger-paper}\\
& A_{bt}-D_{bt}(y)+W_{bt}+G_{bt}(x)=I^+_{bt}(y),
&& \forall b\in\mathcal{B},\ t\in\mathcal{T}^{\mathrm E},
\label{eq:compact-partition-paper}\\
& y_{kb}\in\{0,1\},
&& \forall (k,b)\in\mathcal Y,
\label{eq:compact-binary-paper}\\
& A_{bt},W_{bt}\in\{0,1\},
&& \forall b\in\mathcal B,\ t\in\mathcal T^{\mathrm E},\notag\\
& \sigma_{kbcs}\in\{0,1\},
&& \forall (k,b)\in\mathcal J,\ c\in\mathcal C_b,\ s\in\mathcal S_{kbc},\notag\\
& x_{kbct}\in\{0,1\},
&& \forall (k,b)\in\mathcal J,\ c\in\mathcal C_b,\ t\in\mathcal T^{\mathrm S}.\notag
\end{align}
\end{subequations}
\end{minipage}%
}
\caption{Full five-family MILP (model~F).}
\label{fig:compact-milp-box}
\end{figure*}

\subsection{The full five-family model~F}
\label{sec:full-milp}

Figure~\ref{fig:compact-milp-box} states the no-stockout scheduling problem as a mixed-integer linear program. Because it retains all five binary families, this formulation is termed the \emph{full} model~F. Its rows are organized into functional groups, each corresponding to an operational requirement in Section~\ref{sec:background}.

The service rows~\eqref{eq:compact-assign-paper}--\eqref{eq:compact-serve-paper} encode the no-stockout rule. The assignment row~\eqref{eq:compact-assign-paper} assigns exactly one battery from the event's demand pool $\mathcal{B}_k$, ensuring compatible service for every arrival. The serve row~\eqref{eq:compact-serve-paper}, $D_{bt}(y)\le A_{bt}$, requires that battery to be full and ready before the event epoch. Thus, every issue draws from pre-service full stock, and the two rows jointly provide the formal no-stockout guarantee throughout the planning horizon.

The availability rows~\eqref{eq:compact-init-paper}--\eqref{eq:compact-Adyn-paper} track each battery's ready state over the planning horizon. The initialization~\eqref{eq:compact-init-paper} fixes each battery's full state at the beginning of the planning horizon. The dynamics~\eqref{eq:compact-Adyn-paper} then advance this state according to $A_{b,t+1}=A_{bt}-D_{bt}(y)+\Gamma_{b,t+1}(\sigma)$, removing each issued battery and crediting each recharge completion. The completion is credited at the next epoch, consistent with the timing in Section~\ref{sec:background}.

The job rows~\eqref{eq:compact-start-paper} and~\eqref{eq:compact-duration-paper} construct each recharge from the service that creates it. The start row~\eqref{eq:compact-start-paper}, $\sum_{c}\sum_{s}\sigma_{kbcs}=y_{kb}$, requires exactly one charging start for every assigned service whose return lies in the modeled horizon. The duration row~\eqref{eq:compact-duration-paper}, $x_{kbct}=\sum_{s\le t<s+\tau_{kbc}}\sigma_{kbcs}$, expands that start into the resulting occupancy and reserves the charger for $\tau_{kbc}$ slots, the prescribed charge-to-full upper bound from Section~\ref{sec:background}.

The charger row~\eqref{eq:compact-charger-paper} assigns at most one battery to each charger in each slot, with compatibility incorporated through $\mathcal{C}_b$. The accounting identity below also enforces $G_{bt}(x)\le I^+_{bt}(y)\le1$, limiting each battery to one charger at a time. Finite charger capacity couples recharge decisions across the horizon by allocating each low-price slot among waiting batteries. The start-to-occupancy expansion in~\eqref{eq:compact-duration-paper} and the slotwise capacity constraint~\eqref{eq:compact-charger-paper} follow a standard time-indexed construction for scheduling under TOU prices~\cite{tian2024}, with $y_{kb}$ endogenously activating each charging job.

The accounting rows~\eqref{eq:compact-W-paper} and~\eqref{eq:compact-partition-paper} express physical conservation. The waiting flag is $W_{bt}=\rho_{bt}(y)-\varsigma_{bt}(\sigma)$. The partition row states that $A_{bt}-D_{bt}(y)$, $W_{bt}$, and $G_{bt}(x)$ are battery $b$'s mutually exclusive station-side states after service: ready, waiting, and charging. Their sum is $I^+_{bt}(y)$. Per-battery indexing preserves each physical unit's identity and enforces conservation at the unit level. The displayed rows imply $0\le I^+_{bt}(y)\le1$ and $G_{bt}(x)\le1$, making separate presence and battery-capacity constraints redundant.

The integrality row~\eqref{eq:compact-binary-paper} declares the assignment, charging-start, availability, waiting, and occupancy variables binary, thereby encoding discrete operating states. In particular, integrality of $A$ and $W$ enforces the indivisibility of each physical unit. Section~\ref{sec:reduced} establishes the implied integrality of $A$, $W$, and $x$; Model~R therefore declares only $y$ and $\sigma$ binary.

The objective~\eqref{eq:compact-obj-paper} minimizes the TOU charging cost, $\sum_{t}\sum_{c}\lambda_t e_c\sum x_{kbct}$, by pricing every reserved charger-slot. A feasible point of~\eqref{prob:compact-milp-paper} represents a complete operating plan over the planning horizon, and the optimum is the least-cost plan satisfying the no-stockout requirement.

\begin{figure}[t]
\centering
\setlength{\fboxsep}{4pt}
\fbox{%
\begin{minipage}{\dimexpr\columnwidth-2\fboxsep-2\fboxrule\relax}
\footnotesize
\setlength{\jot}{2pt}
\begin{align}
\min_{y,A,\sigma}\quad & \sum_{(k,b)\in\mathcal{J}}\sum_{c\in\mathcal{C}_b}\sum_{s\in\mathcal{S}_{kbc}} e_c\Big(\sum_{t=s}^{s+\tau_{kbc}-1}\lambda_t\Big)\sigma_{kbcs}
\label{eq:reduced-obj-paper}\\
\text{s.t.}\quad & \text{\eqref{eq:compact-assign-paper}--\eqref{eq:compact-start-paper} and \eqref{eq:compact-charger-paper} with $x=\hat{x}(\sigma)$,}\notag\\
& \text{$y_{kb},\sigma_{kbcs}\in\{0,1\}$ and $A_{bt}\in[0,1]$ for their respective indices.}\notag
\end{align}
\end{minipage}%
}
\caption{Reduced model~R.}
\label{fig:reduced-milp-box}
\end{figure}

\subsection{The reduced model~R}
\label{sec:reduced}

Model~F exposes all five state and decision families in a direct exact transcription. Model~R concentrates the independent binary decisions in the assignment $y$ and charging start $\sigma$, because $(y,\sigma)$ determines the occupancy $x$, waiting flag $W$, and availability $A$. It substitutes out $x$ and $W$ and relaxes $A$ to $[0,1]$, leaving $(y,\sigma)$ explicitly binary. Proposition~\ref{prop:equivalence} establishes that this reduction preserves both the feasible projection and objective values.

Model~R is obtained from F in three steps. First, the duration and waiting rows define $\hat{x}_{kbct}(\sigma)$ and $\hat{W}_{bt}(y,\sigma)$. Let $\hat{G}_{bt}(\sigma)$ denote $G_{bt}$ after substituting $\hat{x}$. This substitution leaves the charger row in $(y,A,\sigma)$ and rewrites the objective so that each start carries the cost of its $\tau_{kbc}$ reserved slots. Second, the algebraically implied partition row is omitted. Third, $A_{bt}$ is relaxed to $[0,1]$, while the remaining structure preserves its integrality. Figure~\ref{fig:reduced-milp-box} states the result.

\emph{Derivation.}\ Fix $b$ and define the pre-service and through-service counts and the cumulative completions by
\begin{equation}
\begin{aligned}
n^-_{bt}(y)&=\!\sum_{\substack{(k,b)\in\mathcal{Y}:a_k<t}}y_{kb},\quad
n_{bt}(y)=\!\sum_{\substack{(k,b)\in\mathcal{Y}:a_k\le t}}y_{kb},\\
H_{bt}(\sigma)&=\sum_{u=1}^{t}\Gamma_{bu}(\sigma),
\end{aligned}
\label{eq:cumulative-proof-paper}
\end{equation}
where $H_{b0}=0$ by the empty-sum convention. Then $D_{bt}=n_{bt}-n^-_{bt}$, and telescoping the availability dynamics gives
$A_{bt}=A_b^0+H_{bt}-n^-_{bt}$. Counting starts not yet completed gives
$\hat{G}_{bt}=\varsigma_{bt}-H_{bt}$. Both identities are algebraic consequences of the model equalities and therefore hold over their continuous relaxations.

For binary $(y,\sigma)$, each $\hat{x}_{kbct}$ is a nonnegative integer. It is bounded by the corresponding job's total number of starts, which is at most one, and is therefore binary. Release feasibility and the start row give $0\le\varsigma_{bt}\le\rho_{bt}$, so $\hat{W}_{bt}=\rho_{bt}-\varsigma_{bt}\ge0$. Moreover, the serve row gives $n_{bt}\le A_b^0+H_{bt}$, while return and completion events satisfy $\rho_{bt}\le n_{bt}$ and $H_{bt}\le\varsigma_{bt}$. Therefore, $\hat{W}_{bt}\le A_b^0\le1$, and $\hat W$ is binary. The closed form for $A$ is integer-valued, so $A_{bt}\in[0,1]$ makes $A$ binary.

Finally,
\begin{equation}
A_{bt}-D_{bt}+\hat{W}_{bt}+\hat{G}_{bt}
=A_b^0+\rho_{bt}-n_{bt}=I^+_{bt}.
\label{eq:partition-derived-paper}
\end{equation}
Thus, the dropped partition row is an identity. Because $\ell_k>0$, every return counted by $\rho_{bt}$ has a corresponding service counted by $n_{bt}$. Therefore, $I^+_{bt}\le A_b^0\le1$. The nonnegative decomposition in~\eqref{eq:partition-derived-paper} gives $I^+_{bt}\ge0$ and $\hat{G}_{bt}\le1$. The omitted presence and battery-capacity inequalities are therefore redundant. Substituting $\hat x$ into F's objective and exchanging its finite sums yields~\eqref{eq:reduced-obj-paper}.

\begin{proposition}[Equivalence of the reduced model]
\label{prop:equivalence}
Under the finite-horizon assumptions above, F and R have the same feasible projection onto $(y,\sigma)$. Every feasible $(y,A,\sigma)$ of R has $A_{bt},\hat{x}_{kbct},\hat{W}_{bt}\in\{0,1\}$. Its extension $(y,A,\hat W,\sigma,\hat x)$ is feasible for F with the same objective value. Consequently, whenever the models are feasible, they attain the same optimal value.
\end{proposition}

\begin{IEEEproof}
The derivation proves integrality of the reconstructed families and all omitted identities and inequalities. Hence, restricting any F-feasible point to $(y,A,\sigma)$ satisfies R. Conversely, extending an R-feasible point by $x=\hat x(\sigma)$ and $W=\hat W(y,\sigma)$ satisfies every row of F. The finite-sum rearrangement above preserves the objective in both directions. The two maps preserve $(y,\sigma)$ and the objective values, which proves the projection statement. For nonempty feasible sets, this correspondence also establishes equality of the attained minima.
\end{IEEEproof}

Model~R reduces the explicit binary dimension from $m_{\mathrm F}=|y|+|\sigma|+|x|+|W|+|A|$ to $m_{\mathrm R}=|y|+|\sigma|$. Specifically, $|x|=|\mathcal{T}^{\mathrm S}|\sum_{(k,b)\in\mathcal{J}}|\mathcal{C}_b|$, $|\sigma|=\sum_{(k,b)\in\mathcal J}\sum_{c\in\mathcal C_b}|\mathcal S_{kbc}|\le|x|$, and $|A|=|W|=|\mathcal B|\,|\mathcal T^{\mathrm E}|$. Model~R removes $x$ and $W$ and makes $A$ continuous. Section~\ref{sec:experiments} quantifies the resulting computational effect.

\section{Price-Guided Battery-Path Decomposition}
\label{sec:price-guided-bpd}

The exact formulations couple battery assignments with charger and start-time decisions across the full time grid. A path-based formulation compresses each battery's temporal decisions into a single column while retaining service coverage and charger-capacity coupling in the master problem, supporting repeated solution as the fleet grows.

P-BPD operates over complete battery paths that compactly encode the exact model's time-indexed charging decisions. Each path comprises a feasible service sequence and its associated charging blocks. A set-partition master selects one path per battery, and a dynamic-programming pricing subproblem generates improving paths with negative reduced cost. P-BPD is a Dantzig--Wolfe/column-generation reformulation in which one column represents a complete feasible history of a single physical battery~\cite{barnhart1998}.

\subsection{Battery paths and path attributes}

For a battery $b$, let $V_b=\{k\in\mathcal{K}:(k,b)\in\mathcal{Y}\}$ be the set of events that battery $b$ can serve. In the canonical path model, a nonempty path begins with a battery satisfying $A_b^0=1$; a battery with $A_b^0=0$ takes the empty path. Each nonempty path records an ordered event sequence and its in-horizon charging decisions:
\begin{equation}
\begin{aligned}
p&=\big((k_1,\ldots,k_m),((c_i,s_i))_{i\in I_p}\big),\\
I_p&=\{i:1\le i<m\}\cup\{m:r_{k_m}\le T_{\max}\}.
\end{aligned}
\end{equation}
where $k_i\in V_b$ and $a_{k_1}<\cdots<a_{k_m}$. The recorded charger and integer start slot satisfy
\begin{equation}
\begin{aligned}
c_i&\in\mathcal{C}_b,\quad s_i\in\mathcal T^{\mathrm S},\quad
r_{k_i}\le s_i,\qquad i\in I_p,\\
s_i+\tau_{k_i b c_i}&\le
\begin{cases}a_{k_{i+1}},&i<m,\\T_{\max},&i=m.
\end{cases}
\end{aligned}
\end{equation}
The empty path is available to every battery. For every served event whose return lies in the modeled horizon, the path retains the corresponding recharge, including the terminal recharge following its final service. This terminal block preserves charger occupancy and cost accounting through the horizon boundary.

Let $\mathcal{P}_b$ be the set of feasible paths for battery $b$. Path indices are local to $b$. For each $p\in\mathcal{P}_b$, define
\begin{align}
\delta_{bpk}&=\begin{cases}
1, & \text{if path }p\text{ serves event }k,\\
0, & \text{otherwise,}
\end{cases}\\
q_{bpct}&=\begin{cases}
1, & \text{if path }p\text{ occupies charger }c\text{ in slot }t,\\
0, & \text{otherwise.}
\end{cases}
\end{align}
The path cost is
\begin{equation}
C_{bp}=\sum_{c\in\mathcal{C}}\sum_{t\in\mathcal{T}^{\mathrm S}}\lambda_t e_c q_{bpct}.
\end{equation}
Complete battery histories have also been used as columns in EV service-network design with charging and battery reuse~\cite{diao2026snd}. P-BPD specializes this column concept to exogenous swap-event coverage at a single station.

\subsection{Path master problem}

Introduce a binary variable $z_{bp}$ that equals one if battery $b$ follows path $p$. The path master problem is
\begin{subequations}
\label{prob:path-master-paper}
\begin{align}
\min_z\quad
& \sum_{b\in\mathcal{B}}\sum_{p\in\mathcal{P}_b}C_{bp} z_{bp}
\label{eq:pm-obj-paper}\\
\mathrm{s.t.}\quad
& \sum_{b\in\mathcal{B}}\sum_{p\in\mathcal{P}_b}\delta_{bpk}z_{bp}=1,
&& \forall k\in\mathcal{K},
\label{eq:pm-cover-paper}\\
& \sum_{p\in\mathcal{P}_b}z_{bp}=1,
&& \forall b\in\mathcal{B},
\label{eq:pm-onepath-paper}\\
& \sum_{b\in\mathcal{B}}\sum_{p\in\mathcal{P}_b}q_{bpct}z_{bp}\le 1,
&& \forall c\in\mathcal{C},\ t\in\mathcal{T}^{\mathrm S},
\label{eq:pm-charger-paper}\\
& z_{bp}\in\{0,1\},
&& \forall b\in\mathcal{B},\ p\in\mathcal{P}_b.
\label{eq:pm-domain-paper}
\end{align}
\end{subequations}
Constraint~\eqref{eq:pm-cover-paper} covers every customer event exactly once. Constraint~\eqref{eq:pm-onepath-paper} selects one physical history for each battery, and constraint~\eqref{eq:pm-charger-paper} enforces charger capacity at slot resolution. Customer-cover and charger-capacity coupling remains in the master, while all single-battery temporal feasibility conditions are embedded in the columns. Per-battery decomposition has also been used for large-scale fleet charging through Lagrangian relaxation~\cite{fotedar2025}. P-BPD keeps event coverage and charger-slot capacity in the master while embedding single-battery temporal feasibility in each column.

\subsection{Restricted master and artificial columns}

The number of feasible paths is exponential in the worst case, so P-BPD works with generated subsets $\widehat{\mathcal{P}}_b\subset\mathcal{P}_b$. The initial pool contains empty paths, singleton paths, and battery paths extracted from greedy feasible schedules. To initialize the restricted LP, artificial variables $u_k\ge0$ replace the cover rows with $\sum_{b,p\in\widehat{\mathcal P}_b}\delta_{bpk}z_{bp}+u_k=1$ and add $M\sum_k u_k$ to the objective for a fixed large penalty $M$. Final-solution acceptance requires $u_k=0$ for every event.

At each column-generation iteration, the linear relaxation of the restricted master is solved. Let $\alpha_k,\beta_b\in\mathbb R$ be the dual multipliers for the customer-cover and one-path equalities. The multiplier of the charger-capacity inequality is written as $-\mu_{ct}\le0$, where $\mu_{ct}\ge0$. The reduced cost of a candidate path $p$ for battery $b$ is
\begin{equation}
\bar{c}_{bp}=C_{bp}-\sum_{k\in\mathcal{K}}\alpha_k\delta_{bpk}
+\sum_{c\in\mathcal{C}}\sum_{t\in\mathcal{T}^{\mathrm S}}\mu_{ct}q_{bpct}-\beta_b.
\label{eq:reduced-cost-paper}
\end{equation}
A path with negative reduced cost is added to the restricted master.

\subsection{Dynamic-programming pricing}

For each battery $b$, pricing seeks a least-reduced-cost path over the predecessor neighborhood $\mathcal{N}_b(\cdot)$ defined below. Because $\beta_b$ is constant across battery $b$'s paths, it suffices to minimize
\begin{equation}
\widetilde{C}_{bp}=C_{bp}-\sum_{k\in\mathcal{K}}\alpha_k\delta_{bpk}
+\sum_{c\in\mathcal{C}}\sum_{t\in\mathcal{T}^{\mathrm S}}\mu_{ct}q_{bpct}.
\end{equation}
Define the dual-adjusted charging-slot cost
\begin{equation}
w_{ct}=\lambda_t e_c+\mu_{ct}.
\end{equation}
If battery $b$ serves event $i$ and then event $j$, the cheapest feasible charge between them has cost
\begin{equation}
\phi^b_{ij}=\min_{c\in\mathcal{C}_b}
\min_{\substack{s\in\mathcal T^{\mathrm S}:r_i\le s,\ s+\tau_{ibc}\le a_j}}
\sum_{t=s}^{s+\tau_{ibc}-1}w_{ct},
\label{eq:transition-cost-paper}
\end{equation}
where the minimum of an empty set is defined as $+\infty$. For fixed $(i,b,c)$, a prefix sum provides the cost of each candidate block in constant time. A sweep over integer starts maintains the least-cost block for successive deadlines.
For each candidate terminal event $j$, let $\phi^b_{j\dagger}$ be the same minimum with deadline $T_{\max}$ when $r_j\le T_{\max}$, and let $\phi^b_{j\dagger}=0$ otherwise. This is the adjusted cost of the required final recharge when its return lies in the horizon.

Sort events lexicographically by $(a_k,k)$ and define the feasible predecessor set
\begin{equation}
E_b(j)=\{i\in V_b:a_i<a_j,\ \phi^b_{ij}<+\infty\}.
\end{equation}
Exhaustive pricing uses $\mathcal N_b(j)=E_b(j)$. Band-limited pricing uses the $L$ elements of $E_b(j)$ with the largest $(a_i,i)$. When $|E_b(j)|<L$, band-limited pricing uses all elements. Let $F_b(j)$ be the minimum dual-adjusted value of a nonempty path ending at $j$, excluding a terminal charge. Under the convention $\min\varnothing=+\infty$, the recursion is
\begin{equation}
F_b(j)=-\alpha_j+
\min\left\{0,
\min_{i\in\mathcal{N}_b(j)}\left(F_b(i)+\phi^b_{ij}\right)
\right\}.
\label{eq:pricing-dp-paper}
\end{equation}
The zero option starts a new path at $j$. Induction in strict arrival order proves that, with exhaustive neighborhoods, $F_b(j)$ equals the minimum adjusted cost among all admissible nonempty paths ending at $j$, before its terminal recharge. Hence, the least reduced cost over both nonempty and empty paths is $\min\{\min_{j\in V_b}[F_b(j)+\phi^b_{j\dagger}],0\}-\beta_b$, where $\min_{j\in\varnothing}F_b(j)=+\infty$. When this minimum is below the negative reduced-cost tolerance, the algorithm reconstructs the minimizing event sequence and stored argmin charging blocks, including any required terminal block, and inserts the resulting path into the restricted master.

\subsection{Computational heuristics and implementation}
\label{sec:computational-policy-paper}

The reported P-BPD implementation combines three speed-oriented devices. Band-limited pricing retains the $L$ most recent feasible elements in each predecessor neighborhood. Cost-only transitions store the exact values $\phi^b_{ij}$ and $\phi^b_{j\dagger}$ with their argmins, reconstructing chargers and starts for each returned path; separability across charging blocks makes this representation lossless. Greedy schedules with randomized tie-breaking seed the initial column pool. Exhaustive pricing continued through convergence supplies the full path master's LP certificate. The experiments use fixed band and iteration budgets and assess the resulting integer schedules against the certified reference solutions and dual-bound optimality-gap guarantees described in Section~\ref{sec:experiments}.

\subsection{Online controllers}
\label{sec:controllers}

Model~R and P-BPD serve as window solvers in a rolling-horizon controller. At each decision epoch, they optimize a finite look-ahead horizon and commit the actions due before the next re-planning event. A source-job augmentation represents the carry-in state. An in-trip or waiting battery contributes an exogenous job with a fixed release epoch and a decision-dependent compatible charging block, while an ongoing charge contributes fixed remaining occupancy and completion. In the path graph, these states become source-to-first-event arcs with the same cost, occupancy, and ready-time data. Batteries ready at the re-plan epoch have $A_b^0=1$. The augmentation adds identical fixed or source-job terms to F and R, and Proposition~\ref{prop:equivalence} continues to apply to their endogenous $(y,\sigma)$ block. Each window model assigns a high penalty to an artificial no-service option, prioritizing service in the restricted solve; the controller records any selected no-service option as a committed stockout.

\section{Experiments}
\label{sec:experiments}

The experiments quantify the trade-offs among charging cost, service, and computation time on synthetic instances and in a recorded-demand replay under the deployed station configuration. Models F and R share the same optimum, and each run that reaches the stated relative MIP gap certifies its incumbent. P-BPD solution quality is evaluated against certified references or available exact-solver dual bounds.

Every experiment is a simulation on the synthetic instances of Section~\ref{sec:synthetic-gen} or the recorded-demand instances of Section~\ref{sec:real-data}; rolling experiments use a 24-hour look-ahead horizon. All runs price slots with KEPCO's low-voltage, option-II summer TOU tariff---off-, mid-, and peak rates~\cite{KEPCO_ElectricityTariff} of $83.1$, $140.0$, and $270.8$~KRW/kWh over the bands of Fig.~\ref{fig:station}(b)---so cost differences between methods reflect scheduling alone. A slot at calendar origin $h_0$ and duration $\Delta$ hours is priced at hour $h(t)=(h_0+t\Delta)\bmod24$, which reduces to $24t/T$ for the midnight-origin synthetic instances. All solves use HiGHS~1.14.0~\cite{huangfu2018parallelizing} on an Apple~M1 (8 cores, 8\,GB RAM, macOS~14.0), single-threaded unless otherwise noted. %

\subsection{Compared methods}
\label{sec:methods-compared}

Five methods are compared: the exact formulation F, its equivalent reduction R (Section~\ref{sec:methods}), P-BPD (Section~\ref{sec:price-guided-bpd}), and two baselines. \emph{ASAP} charges each returned battery at the earliest free compatible slot; the \emph{TOU baseline} chooses the cheapest feasible slot under a per-pool deadline that keeps the pool's running ready count nonnegative. Both baselines order jobs by return epoch and use the lowest-indexed compatible battery to determine each recharge duration. When no feasible start remains, the affected swap is counted as unserved, so the reported service rate directly quantifies charger-capacity conflicts.

\emph{Solver configuration.} F and R are solved single-threaded by HiGHS~1.14.0 through SciPy~1.17.1's \texttt{milp} at relative MIP gap $10^{-4}$, with a wall-clock budget of $600$\,s at $B\le18$ and for R at $B{=}36$, and $1800$\,s otherwise. An F or R run is recorded as certified when it reaches the target gap within its wall-clock budget. P-BPD uses predecessor-band size $L=5$, initializes the restricted master with empty paths, singleton paths, and paths extracted from $8$ greedy schedules, performs at most $10$ restricted-LP pricing iterations, and adds paths below reduced-cost tolerance $-10^{-7}$. It then solves the restricted integer master at relative gap $10^{-4}$ offline or $10^{-2}$ in rolling re-plans. The solver gap certifies F and R over their full formulations; for P-BPD, it certifies optimization over the current restricted column set. The comparisons below assess P-BPD's full-model quality against certified references or exact-solver dual bounds.

\subsection{Evaluation metrics and protocol}
\label{sec:metrics-protocol}

A controller's service is summarized at two resolutions. Let $\mathbb{1}[\cdot]$ denote the indicator. The \emph{day reliability} $R_{\mathrm{day}}$ is the fraction of evaluated days served with no stockout,
\begin{equation}
R_{\mathrm{day}}=\frac{1}{N}\sum_{d=1}^{N}\mathbb{1}\!\left[\text{day }d\text{ served with no stockout}\right],
\label{eq:day-reliability}
\end{equation}
where $N\in\mathbb Z_{++}$ is the number of evaluated days. A day contributes to $R_{\mathrm{day}}$ precisely when the assignment and service conditions \eqref{eq:compact-assign-paper}--\eqref{eq:compact-serve-paper} hold for every committed step, so $R_{\mathrm{day}}$ measures complete daily service. The \emph{swap service rate} $S$ is the fraction of individual swaps served from full stock,
\begin{equation}
S=\frac{1}{K}\sum_{k=1}^{K}\mathbb{1}\!\left[\text{swap }k\text{ served from full stock}\right],
\label{eq:swap-service}
\end{equation}
with $K\in\mathbb Z_{++}$ the number of unique committed swaps over the evaluation interval. Each day contributes equally to $R_{\mathrm{day}}$, whereas each swap contributes weight $1/K$ to $S$. A feasible static-horizon plan serves every represented event, and a rolling controller attains $R_{\mathrm{day}}=S=1$ when every committed service succeeds. Because replayed swaps are temporally dependent within and across days, the two empirical measures report service at complementary day and event resolutions.

The reported cost is the energy component of the bill. Under the station's low-voltage tariff, the $2{,}390$~KRW/kW demand charge is levied on contracted power, which the installed chargers fix identically for every method; the comparisons therefore isolate the schedule-dependent energy component. A common terminal-recharge convention charges every return represented in the instance, including the final return of each battery path, giving every method the same terminal-inventory treatment at $T_{\max}$. On instances with positive denominators, the cost premium is $(\text{cost}-\text{reference})/\text{reference}$ and the reduction against a baseline is $(\text{baseline}-\text{cost})/\text{baseline}$. Solve time is the wall-clock time of one solve: a finite-horizon certification in the synthetic scaling study and one window re-plan online.

A synthetic instance spans one simulated day, while the formulation accepts any finite horizon. The reduction, ablation, Small, and Medium studies use $30$ seeds per regime (cases $0$--$29$); Section~\ref{sec:frontier} specifies the five-seed design for the Large and xlarge regimes.

\subsection{Synthetic experiments}
\label{sec:synthetic}

\subsubsection{Synthetic instance generation}
\label{sec:synthetic-gen}

A \emph{regime} fixes a station size. Table~\ref{tab:regime-summary-paper} lists the four scale-study regimes ($B=18,36,54,72$). An additional $B{=}8$ regime supports the reduction comparison and ablation, with its configuration given in Section~\ref{sec:cert-frontier}.

\emph{Station structure.} The $B=|\mathcal{B}|$ batteries and $C=|\mathcal{C}|$ chargers are each shuffled and assigned cyclically to the regime's compatibility classes, with the assignment ensuring that every class has at least one charger. A battery recharges only on a charger of its class. In these instances, each class corresponds to a single demand pool.
Charger $c$ has speed $s_c\sim\mathcal{U}(1.15,2.40)$ and draws energy $e_c=0.8+0.55\,s_c$ per slot. Battery $b$ has duration factor $f_b\sim\mathcal{U}(0.85,1.20)$, and all batteries start full.

\emph{Events.} Each event $k$ draws a charging need $\nu_k\sim\mathcal{U}(\nu_{\min},\nu_{\max})$ (abstract units) and an integer return delay $\ell_k\sim\mathcal{U}\{\ell_{\min},\ell_{\max}\}$ from the ranges in Table~\ref{tab:regime-summary-paper}. Its battery returns $\ell_k$ slots after issue and then recharges on a compatible charger $c$ in $\tau_{kbc}=\max(1,\lceil \nu_k f_b/s_c\rceil)$ slots, with incompatible pairings barred. The models receive the arrivals, return delays, pools, and $\tau$; the charging need $\nu_k$ is represented through $\tau$. The per-battery load $K/B$ ($K=|\mathcal{K}|$) is near $3$ in every tabulated regime ($2.78$ at $B{=}18$) and $2.25$ in the additional $B{=}8$ regime.

\emph{Feasibility by construction.} Arrival epochs are generated jointly with a witness schedule: at each generator step, an available battery is assigned an event no earlier than its ready epoch and, on return, is placed on the compatible charger with the earliest feasible finish. The generator retains events whose assignment and recharge finish within the horizon. The witness fixes the event times and remains separate from the evaluated solver inputs. Every retained synthetic instance therefore admits at least one feasible no-stockout schedule, allowing the experiments to quantify each method's ability to recover full service on known-feasible instances.

\begin{table}[t]
\centering
\caption{Synthetic instance regimes.}
\label{tab:regime-summary-paper}
\begin{tabular}{lrrrrccc}
\toprule
Regime & $B$ & $C$ & $K$ & $T$ & classes & $\ell$ & need \\
\midrule
Small   & 18 & 5  & 50  & 72  & 3 & 1--5 & 2.5--8.0 \\
Medium  & 36 & 8  & 110 & 120 & 4 & 1--6 & 2.5--8.0 \\
Large   & 54 & 12 & 165 & 120 & 4 & 1--6 & 2.5--8.0 \\
xlarge  & 72 & 16 & 220 & 120 & 4 & 1--6 & 2.5--8.0 \\
\bottomrule
\end{tabular}
\end{table}

\subsubsection{The reduction's certified speedup}
\label{sec:cert-frontier}

The certified objective values of F and R agree to numerical tolerance on every certified paired run, as established by Proposition~\ref{prop:equivalence}. The ratios of median F to median R certification times are $3.1$, $5.6$, and $7.2$ at $B=8$, $18$, and $36$, respectively (Fig.~\ref{fig:reduction}). The $B{=}8$ setting ($C{=}3$, $K{=}18$, $T{=}36$, $2$ classes, $\ell$ $1$--$4$, need $2.5$--$7.0$) is the additional regime used for the reduction study and ablation.

\begin{figure}[t]
\centering
\includegraphics[width=\columnwidth]{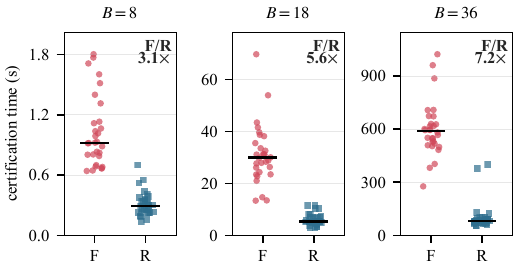}
\caption{Per-instance certification times for F and R ($n{=}30$ at $B{=}8$, $18$, and $36$). Horizontal lines mark medians; panel labels give the F-to-R ratio of medians.}
\label{fig:reduction}
\end{figure}

\subsubsection{Solve time, cost, and service across scale}
\label{sec:frontier}

The Small ($B{=}18$) and Medium ($B{=}36$) regimes use $30$ seeds each; the Large ($B{=}54$) and xlarge ($B{=}72$) regimes use $5$ seeds each.

On the 30-seed regimes (Fig.~\ref{fig:scoreboard}), F and R obtain a $10^{-4}$ relative MIP certificate on every seed, and P-BPD returns a feasible full-service schedule on every seed, with median solve times of $0.24$--$1.4$~s and median premiums of $7\%$ (Small) and $8\%$ (Medium) over the certified reference. Service metrics therefore distinguish the two operational baselines on these known-feasible instances.

Both baselines solve in under $0.05$~s. For each baseline, the certified reference reduces cost by $50$--$60\%$ on fully served instances with an available certified reference, corresponding to ASAP and TOU premiums of $131$--$152\%$ and $100$--$111\%$, respectively. In the Small regime, ASAP attains $80\%$ day reliability and $99.9\%$ swap service, while TOU attains $60\%$ and $99.1\%$, respectively.

The two larger regimes characterize exact-certification scalability. Model size reaches approximately $0.5$, $1.5$, and $3.5$ million binaries at Medium, Large, and xlarge, respectively. At both Large and xlarge, F reaches the time budget without certification on all $5$ seeds, while R certifies $4$ of $5$ and $1$ of $5$ seeds, respectively. P-BPD returns a schedule on every seed in both regimes, with median times of $4.9$ and $8.0$~s. Each plotted premium uses a certified reference when available; for the remaining instances, a positive exact-solver dual bound $L$ and P-BPD cost $H$ establish $0<L\le\mathrm{OPT}\le H$ and hence $(H-\mathrm{OPT})/\mathrm{OPT}\le(H-L)/L$. These certificates yield ex post optimality-gap upper bounds of $9.2$--$10.9\%$ at Large and $9.0$--$11.9\%$ at xlarge.

\begin{figure*}[t]
\centering
\includegraphics[width=0.8\textwidth]{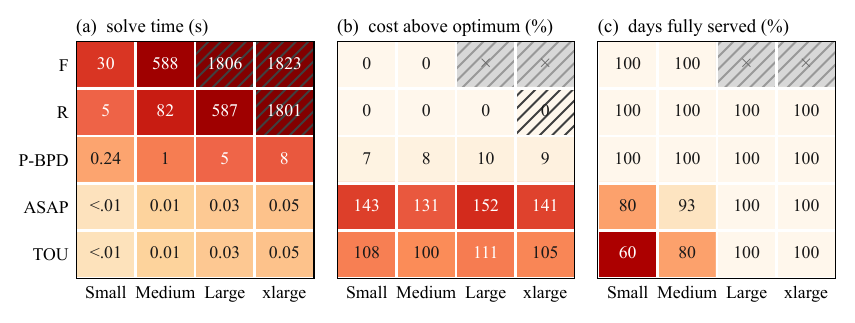}
\caption{Method comparison under the real summer tariff: (a) median solve time, (b) median cost comparison, and (c) percent of days fully served, by method (rows) and station size (columns; $n{=}30$ seeds at Small and Medium, $n{=}5$ at Large and xlarge). Panel (b) reports premiums over certified reference costs when available and, for P-BPD on the remaining instances, certified upper bounds on its relative optimality gap. Baseline cost cells use fully served instances with an available certified reference. Hatching marks exact solves stopped at the time budget; their times are right-censored. $\times$ denotes unavailable comparisons.}
\label{fig:scoreboard}
\end{figure*}

\subsubsection{Ablation of the P-BPD heuristic}
\label{sec:ablation}

Each of P-BPD's three computational devices (Section~\ref{sec:computational-policy-paper}) is removed one at a time and compared with the full method on $90$ instances ($30$ seeds at each of $B{=}8$, $18$, and $36$) under the real summer tariff (Table~\ref{tab:ablation}). Each device contributes a distinct computational benefit. Cost-only transitions reduce median time from $0.39$ to $0.24$~s at unchanged cost, and band-limited pricing reduces the observed maximum from $14.2$ to $5.9$~s while preserving the median. Greedy warm starts increase the plan-return count from $87$ to $89$ of $90$ and reduce cost by approximately $1\%$, with a median-time change from $0.16$ to $0.24$~s.

\begin{table}[h]
\centering
\caption{Leave-one-out ablation of P-BPD's three devices on $90$ instances under the real summer tariff. Runtime summaries cover all runs; $\Delta$cost is the median paired percentage change over instances for which both the variant and full method return plans. The full method is the reference.}
\label{tab:ablation}
\begin{tabular}{lrrrr}
\toprule
                          & Median   & Max      & Plans found & $\Delta$cost \\
Variant                   & time (s) & time (s) & (/90)       & (\%) \\
\midrule
Full method               & 0.24 & 5.9  & 89 & --- \\
$-$\,cost-only transition & 0.39 & 7.7  & 89 & $\approx 0$ \\
$-$\,band-limited pricing & 0.24 & 14.2 & 89 & $\approx 0$ \\
$-$\,greedy warm starts   & 0.16 & 2.2  & 87 & $+1.0$ \\
\bottomrule
\end{tabular}
\end{table}

\subsection{Recorded-demand replay and deployment-oriented scaling}
\label{sec:real-anyang}

This section replaces the synthetic generator with the deployed station's configuration and recorded demand, replaying each method at the recorded operating point and tracing its performance as the station scales. Operation is event-triggered under the connected-fleet setting of Section~\ref{sec:background}: at each swap arrival or charge completion, the controller re-solves a 24-hour sliding window from the current planned arrival and return schedule and executes actions due before the next event. The replay study (Section~\ref{sec:live-admissibility}) supplies the recorded June sequence at the station's real size as this schedule; the station-scaling study (Section~\ref{sec:capacity}) samples demand from the same log at scale factors $f$. Model~R and P-BPD operate on the same demand stream and initial station configuration under a common controller protocol.

\subsubsection{Station configuration and recorded demand}
\label{sec:real-data}
\label{sec:real-demand}

The station configuration is a commercial battery-swap service for electric taxis in Anyang, Gyeonggi Province, South Korea. Four customer pools ($84.0$, $64.0$, $72.0$, $77.4$~kWh) span three compatibility classes, the $72.0$ and $77.4$~kWh pools sharing one, and six chargers serve them: $100$/$50$/$100$~kW for the $84.0$~kWh class, one $50$~kW for the $64.0$~kWh class, and two $50$~kW for the shared class, the second a dual-generation bay serving both pools. Each charger draws $0.2P$~kWh in each occupied $12$-min slot at rated power $P$~kW, and a fixed buffer of $7$/$10$/$2$/$6$ full spares backs the pools. A capacity-$\kappa$ battery on a power-$P$ charger recharges in $\lceil \kappa\cdot0.58\cdot60\,\eta_P/(12P)\rceil$ slots, with charging-overhead factors $\eta_{50}{=}1.270$ and $\eta_{100}{=}1.604$ and $0.58$ the assumed summer return depth of discharge. %

Each active vehicle holds one battery in circulation, so $B$ is the shelf buffer plus one battery per active, distinct-in-window vehicle in each pool. The warm-up interval generates the carry-in state at the start of reporting, and the trailing look-ahead interval is excluded from the reported metrics.

The replay uses every completed swap recorded at the station during June~2026 (June~1--30) in the four deployed customer pools: $1{,}002$ swaps by $37$ vehicles over $30$ full days, a mean of $33.4$ per day, split $609$/$215$/$128$/$50$ across the $84$/$77.4$/$64$/$72$~kWh pools. Each record gives the swap's start and end times in Korea Standard Time, the vehicle, and the battery capacity taken; personal fields are dropped, and each vehicle is relabeled with an anonymous identifier. Each record is one event, and the dispatched battery returns depleted at that vehicle's next recorded swap, which sets its return delay; a vehicle's last swap in the horizon has no in-window return. The recorded-demand replay evaluates reliability directly on observed demand, complementing the known-feasible synthetic instances. %

\subsubsection{Rolling-horizon admissibility on real demand}
\label{sec:live-admissibility}

The station's recorded June~2026 swap log (Section~\ref{sec:real-demand}) is replayed as a continuous multi-day stream under the connected-fleet schedule assumption. R and P-BPD follow the same recorded demand and controller protocol and use a relative solver gap of $10^{-2}$ for their respective integer models. At each re-plan, the next 24 hours of realized arrivals and returns are supplied as planning input, isolating rolling-scheduler performance under perfect information. Across $1{,}118$ R windows and $1{,}064$ P-BPD windows, both controllers return plans within the $120$\,s budget, and fallback usage is zero. %
Implementation tests validate the rolling state transition against an independently coded window engine and enumerated small instances against a brute-force oracle.
The reported interval spans the $30$ full June days, with cost compared per day.

The station's real size is comparable to the Large and xlarge rungs of the synthetic scaling study (Section~\ref{sec:frontier}). Its observed operating intensity is approximately $33$ swaps per day across $62$ batteries, $K/B\approx0.5$, or one-sixth of the synthetic $K/B\approx3$. The replay evaluates reduced-model re-plan latency at this asset scale under observed operating intensity. %

On real demand, the reduced-model re-plan fits the per-window compute budget: over the $30$ days, R returns each window in a mean of $3.2$\,s ($95$th percentile $7.9$\,s, maximum $19.6$\,s), and every committed swap is served. P-BPD returns every window in under $2.2$\,s. Thus, both controllers remain within the $120$\,s budget on every window, and all $30$ reported days have no stockout.
Over the same recorded-demand stream, the aggregate committed costs of R and P-BPD agree within $0.12\%$. Both reduce cost by approximately $50\%$ relative to ASAP and by $1.3\%$ relative to the TOU baseline.

\subsubsection{The reduced-model-to-heuristic crossover at scale}
\label{sec:capacity}

\begin{figure}[t]
\centering
\includegraphics[width=\columnwidth]{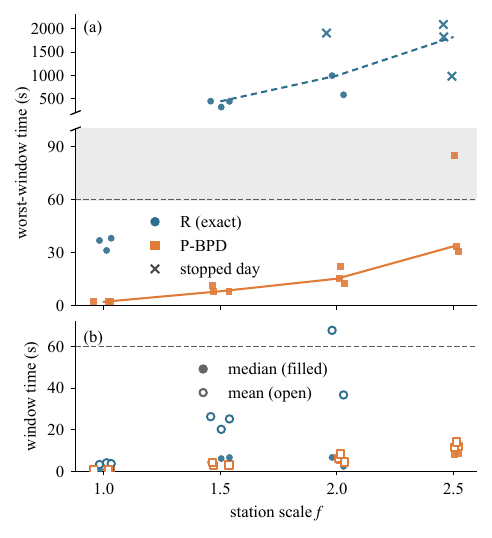}
\caption{Re-plan times relative to the $60$\,s operational target (dashed; exceedance region shaded in (a)): (a)~per-day maximum times, with an axis break at $100$\,s and per-rung median lines; (b)~per-day medians and means for completed days. Crosses mark days ending at the screen's first stopping event (a solve above $D$ or a right-censored solve); their recorded maxima and the R trace for $f{\ge}2$ are lower bounds.}%
\label{fig:crossover}
\end{figure}

The station-scaling study examines the observed computational crossover on days sampled from the June log. The resampling unit is the \emph{car-day}, one vehicle's complete recorded swap sequence on one June date: car-days are drawn with replacement from the full June set of their battery pool, and every drawn arrival time receives independent uniform jitter of $\pm12$ minutes before slotting. For each pool $p$, an instance-day draws $\mathrm{round}(f\,M_p)$ car-days, where $M_p$ is the pool's car-day count on the busiest June day ($18$/$5$/$1$/$8$ for the $84.0$/$64.0$/$72.0$/$77.4$~kWh pools; $32$ car-days, approximately $40$ swaps). An instance comprises a sampled warm-up day, reported day, and look-ahead day at a common factor $f\in\{1.0,1.5,2.0,2.5\}$; each rung is evaluated on three independently sampled instances. Chargers scale as $\mathrm{round}(6f)$ by cycling the real type tuple, and spares as $\mathrm{round}(f\cdot\mathrm{buffer})$ per pool. Each sampled car-day represents a distinct synthetic vehicle carrying one battery, and vehicles do not recur across days. The resulting roster-scaled populations contain $121$/$184$/$242$/$300$ batteries across the rungs and represent roster scaling, distinct from whole-station replication. %

R and P-BPD follow the same sampled demand streams, initial configurations, and re-planning protocol at the operational gap $10^{-2}$, with fallback disabled. Solves use a nominal $1800$\,s limit; because wall-clock termination is recorded when the solver returns, recorded times can exceed the nominal limit, and unfinished solves are right-censored. For a reported day with $K$ swaps, $D=86400\,\mathrm{s}/K$ defines a nominal per-swap compute target. Because controllers also re-plan at charge completions, $D$ serves as a normalized stress-screen threshold rather than a full serial-throughput budget. The screen applies $D$ uniformly to every observed window and stops at the first solve above $D$ or the first right-censored solve. It includes warm-up windows, and all four recorded stopping events occur during warm-up. %

P-BPD satisfies the deadline-$D$ screen on all $12$ sampled days, with a maximum window time of $84.8$\,s. R's pass counts are $3$ of $3$, $3$ of $3$, $2$ of $3$, and $0$ of $3$ across $f=1.0$, $1.5$, $2.0$, and $2.5$, respectively. At $f{=}2.5$, two days are right-censored at the nominal limit, and the third completes a $983$\,s solve against $D=847$\,s. Relative to the $60$\,s operational target, R achieves full-rung coverage at $f{=}1.0$ (maximum window $38$\,s), while P-BPD meets the target on $11$ of $12$ days. Window models grow from approximately $10^{5}$ to $1.4\times10^{6}$ solver variables across the rungs. %

Tail windows determine the screen outcomes. Among windows observed before each stopping point, pooled median solve times remain between $0.8$ and $5.8$\,s at every rung, while the day-maximum window on all $8$ completed R days occurs before the morning demand peak, when the uncommitted demand set and window model are largest. R windows above $100$\,s number $0$ of $272$ solved windows at $f{=}1.0$, $25$ of $365$ at $f{=}1.5$, and $39$ of $279$ at $f{=}2.0$; the recorded maximum grows from $38$ to $447$ to $1912$\,s, with the last value right-censored at the nominal $1800$\,s limit. Across the scale range, $D$ changes from approximately $2160$ to $850$\,s. At $f{=}1.5$, approximately eight windows per sampled day exceed $100$\,s and therefore lie above the $60$\,s target.%

\section{Conclusion}
\label{sec:conclusion}

This paper formulated no-stockout, TOU-aware charging as an exact per-battery MILP and proved its equivalence to Model~R, which concentrates binary decisions in service assignments and charging starts. In the paired reduction study, Model~R lowered median exact-certification time by factors of $3.1$--$7.2$. In the synthetic scaling study, P-BPD returned a full-service schedule for every instance; regime-level median solve times ranged from $0.24$ to $8.0$\,s, with $7$--$8\%$ median cost premiums over certified reference costs at the Small and Medium scales and $9$--$12\%$ certified ex post optimality-gap upper bounds at the Large and xlarge scales. In the 30-day perfect-information station replay, both rolling controllers served all $1{,}002$ recorded swaps while reducing charging-energy cost by approximately $50\%$ relative to immediate charging and by $1.3\%$ relative to the TOU baseline. Together, the exact reduction and price-guided path heuristic provide a scalable scheduling framework that enforces event-level service in every accepted schedule at physical-battery resolution.

\bibliographystyle{IEEEtran}
\bibliography{references_shared,references}

\end{document}